\documentclass[
    aps,
    pra,
    reprint,
    amsmath,
    amssymb,
    superscriptaddress,
    longbibliography
]{revtex4-2}

\usepackage{graphicx}
\usepackage{booktabs}
\usepackage{bm}
\usepackage{hyperref}
\usepackage{mathtools}
\usepackage{subcaption}
\usepackage{float}
\usepackage{amsthm}
\usepackage{dsfont}

\hypersetup{
    colorlinks=true,
    linkcolor=blue,
    citecolor=blue,
    urlcolor=blue
}

\newcommand{\Tr}{\operatorname{Tr}}
\newcommand{\Hil}{\mathcal{H}}
\newcommand{\cS}{\mathcal{S}}
\newcommand{\cM}{\mathcal{M}}
\newcommand{\cP}{\mathcal{P}}
\newcommand{\id}{I}

\newtheorem{theorem}{Theorem}

\newtheorem{proposition}{Proposition}
\newtheorem{definition}{Definition}

\begin{document}

\title{Robust multi-hypothesis quantum-state discrimination under unknown common unitary perturbations via least favorable priors}

\author{Daichi Fujiki}
\email{fujiki@sigmath.es.osaka-u.ac.jp}
\affiliation{Graduate School of Engineering Science, The University of Osaka, 1-3 Machikaneyama, Toyonaka, Osaka 560-8531, Japan}

\author{Fuyuhiko Tanaka}
\email{ftanaka.celas@osaka-u.ac.jp}
\affiliation{Center for Education in Liberal Arts and Sciences, The University of Osaka, 1-16 Machikaneyama, Toyonaka, Osaka 560-0043, Japan}
\affiliation{Center for Quantum Information and Quantum Biology, The University of Osaka, 1-2 Machikaneyama, Toyonaka, Osaka 5600043, Japan}

\date{\today}

\begin{abstract}
We study robust K-ary quantum-state discrimination when all candidate states are affected by the same unknown common unitary perturbation.
The unknown perturbation does not represent the label to be identified, but acts as a nuisance factor that changes the performance of a fixed measurement.
We formulate the problem as a minimax decision problem over the possible perturbations and propose using the Bayes-optimal collective measurement associated with a least favorable prior (LFP) on the nuisance-parameter space.
For a finite discretization of this space, we show that the LFP can be computed by a semidefinite program and that the corresponding value coincides with the finite-grid minimax success probability.
As numerical demonstrations, we consider a binary nonorthogonal qubit model and a nonorthogonal three-state qutrit model with an unknown common unitary perturbation.
The LFP-based measurement substantially flattens the success-probability profile and improves the worst-case success probability compared with a reference-point optimal measurement and a uniform-prior Bayes measurement.
The resulting LFP concentrates its weight on regions of the nuisance-parameter space that actively limit the robust discrimination performance, thereby providing both a constructive measurement design and a diagnostic description of the difficult nuisance-parameter regimes.
\end{abstract}

\maketitle

%%%%%%%%%%%%%%%%%%%%%%%%%%%%%%%%%%%%%%%%%%%%%%%%%%%%%%%%%%%%%%%%%%%%%%%%%%
\section{Introduction}
\label{sec:introduction}
%%%%%%%%%%%%%%%%%%%%%%%%%%%%%%%%%%%%%%%%%%%%%%%%%%%%%%%%%%%%%%%%%%%%%%%%%%

Quantum-state discrimination is a central problem in quantum information theory, with applications ranging from quantum communication and cryptography to metrology and sensing \cite{chefles2000,barnett2009,bae2015,holevo1982}.
Given a finite ensemble of candidate states $\{\rho_a\}_{a=1}^K$ and prior probabilities, the measurement maximizing the average discrimination success probability is characterized by the theory of minimum-error quantum detection developed by Helstrom, Holevo, and Yuen--Kennedy--Lax \cite{helstrom1976,holevo1973,yuen1975,watrous2018}.
The resulting optimization problem admits a semidefinite-programming formulation, and its optimality conditions can be characterized through convex duality \cite{eldar2003,nakahira2015}.

In realistic experiments, however, the candidate ensemble is often not known exactly.
Examples include unknown phase drifts in optical communication, misalignment of reference frames \cite{bartlett2007}, calibration errors in measurement devices, slowly varying unitary drifts, and common rotations induced by a communication channel.
A simple and important model of such uncertainty is
\begin{equation*}
    \rho_a
    \longmapsto
    \rho_a(\theta)
    =
    U(\theta)\rho_a U(\theta)^\dagger ,
    \qquad
    \theta\in\Theta ,
    \label{eq:intro_common_unitary}
\end{equation*}
where the same unknown unitary perturbation acts on all candidate states.
The parameter $\theta$ is not itself the object to be estimated.
Such a parameter is referred to as a nuisance parameter~\cite{suzuki2020}.
It affects the discrimination performance of a fixed measurement, but the decision to be made is still the candidate label $a$.
Problems involving unknown symmetry transformations and covariant quantum measurements, which are closely related to the common-unitary uncertainty considered here, have been studied in quantum estimation and decision theory
\cite{holevo1979covariant,chiribella2005}.

This distinction is important.
The task considered here is not the joint identification of $(a,\theta)$ and not a composite hypothesis test in which $\theta$ labels different hypotheses.
Instead, the observer must choose a single measurement, without knowing $\theta$, that identifies $a$ reliably for all relevant values of $\theta$.
This naturally leads to a minimax criterion:
one seeks a measurement whose worst-case risk over the nuisance-parameter space is as small as possible.

In classical statistical decision theory, minimax procedures and least favorable priors are standard tools for robust estimation and testing \cite{wald1950,ferguson1967,berger1985,lecam1986,huber2009}.
Analogous ideas have also been studied in quantum decision theory. Minimax formulations have been developed for quantum-state discrimination, discrimination with inconclusive outcomes, and quantum-channel discrimination \cite{dariano2005,nakahira2013minimax,dariano2005pauli}. More general quantum statistical decision-theoretic results, including minimax theorems and least favorable priors, have also been established \cite{tanaka2014,Luczak2025Quantum}.
Hunt--Stein-type arguments and symmetry-based reductions provide another route to minimax solutions in quantum statistical problems with nuisance parameters \cite{kumagai2013}. 
In quantum estimation theory, nuisance parameters have been studied from the viewpoint of attainable precision and information geometry \cite{suzuki2019,suzuki2020}.

Robust quantum measurement design has also been investigated in specific discrimination settings, including binary coherent-state discrimination in the presence of experimental imperfections \cite{dimario2018}.
The present work differs from these existing directions in its focus and construction.
We do not optimize over the prior probabilities of the hypotheses themselves.
Instead, the prior is placed on the nuisance parameter $\theta$ describing an unknown common unitary perturbation.
The resulting least favorable prior identifies which perturbations make the discrimination problem most difficult.
We then use the Bayes-optimal measurement for this prior as a robust collective measurement for $K$-ary state discrimination.
This gives a constructive and computationally accessible method for designing robust measurements under unknown common unitary perturbations.

The main contribution of this paper is not the minimax--least-favorable-prior
correspondence itself, which is a fundamental principle in statistical
decision theory, but its concrete use as a computational and diagnostic
framework for robust multi-hypothesis quantum-state discrimination with
nuisance parameters. In particular, we place the least favorable prior on the
unknown nuisance parameter, rather than on the hypotheses to be discriminated.
This distinction is essential: the label prior describes how often each
candidate state is prepared, whereas the nuisance prior identifies which
external perturbations make the fixed-label discrimination problem most
difficult.

More specifically, our contributions are as follows.
First, we formulate robust $K$-ary quantum-state discrimination under an unknown common unitary perturbation as a finite-grid minimax problem over the nuisance parameter.
Second, we show that, on the finite grid, the least favorable prior over the nuisance parameter can be computed through a semidefinite program, and that a Bayes-optimal measurement for this prior gives a minimax measurement when the saddle-point conditions are satisfied.
Third, we demonstrate the method on binary nonorthogonal qubit and three-state qutrit models, showing that the LFP-based collective measurement raises the worst-case success probability and flattens the success-probability profile relative to both a reference-point optimal measurement and a uniform-prior Bayes measurement.

The proposed framework therefore has two complementary roles. As a measurement
design method, it converts a minimax robust-discrimination problem into a
Bayes measurement problem for a suitably chosen nuisance prior. As a diagnostic
tool, the support of the least favorable prior reveals the parameter regions
that actively limit robust discrimination performance. This diagnostic
interpretation is particularly useful in multi-hypothesis and nonorthogonal
models, where the most difficult nuisance values need not be apparent from
the magnitude of the perturbation alone.

The remainder of the paper is organized as follows.
Section~\ref{sec:formulation} introduces the robust discrimination problem.
Section~\ref{sec:lfp_minimax} proves the relation between least favorable priors and minimax measurements.
Section~\ref{sec:numerical_results} defines the measurement strategies compared in the numerical study and presents the qubit and qutrit examples.
Section~\ref{sec:discussion} discusses the interpretation and limitations of the method.
Section~\ref{sec:conclusion} concludes the paper.

%%%%%%%%%%%%%%%%%%%%%%%%%%%%%%%%%%%%%%%%%%%%%%%%%%%%%%%%%%%%%%%%%%%%%%%%%%
\section{Problem formulation}
\label{sec:formulation}
%%%%%%%%%%%%%%%%%%%%%%%%%%%%%%%%%%%%%%%%%%%%%%%%%%%%%%%%%%%%%%%%%%%%%%%%%%

The optimal discrimination of a finite set of known quantum states is a
well-established problem, and necessary and sufficient conditions for
minimum-error measurements are known~\cite{yuen1975}.

In practical settings, however, the states available to the observer may be
affected by unknown physical perturbations, so that a measurement optimized
for the unperturbed states can become sensitive to such uncertainty.
In this section, we consider quantum-state discrimination when all candidate
states are subject to a common unknown unitary perturbation and formulate the
resulting uncertainty as a nuisance parameter.

Our objective is to construct a measurement whose discrimination performance
is robust against variations in this nuisance parameter.
To this end, we propose a Bayes-optimal measurement associated with a least
favorable prior (LFP) over the nuisance-parameter space.
The central idea is to choose the nuisance prior so that the corresponding
Bayes-optimal measurement is also optimal in the worst-case sense.
To establish this connection rigorously, we first formulate the discrimination
problem as a minimax decision problem, introduce the Bayes risk with respect
to a nuisance prior, and then develop a finite-grid formulation that enables
the LFP and the associated measurement to be computed.

\subsection{State family with an unknown common unitary perturbation}

Let $\Hil$ be a finite-dimensional Hilbert space and let
\begin{equation*}
    \{\rho_a\}_{a=1}^K \subset \cS(\Hil)
\end{equation*}
be the candidate quantum states, where
\begin{equation*}
    \cS(\Hil)
    :=
    \left\{
        \rho : \Hil \to \Hil
        \,\middle|\,
        \rho \succeq 0,\;
        \Tr \rho = 1
    \right\}
\end{equation*}
denotes the set of quantum states on $\Hil$.

A measurement optimized for unperturbed states may become suboptimal in the presence of unknown physical perturbations. We therefore consider quantum-state discrimination in which all candidate states are affected by the same unknown unitary perturbation,
\begin{equation*}
    \rho_a(\theta)
    =
    U(\theta) \rho_a U(\theta)^\dagger,
    \qquad
    a=1,\ldots,K,
    \quad
    \theta \in \Theta.
    \label{eq:perturbed_state}
\end{equation*}
Here $\Theta$ is the nuisance-parameter space.
Throughout the theoretical discussion, $\Theta$ is assumed to be compact
when continuous; in the numerical construction it is replaced by a finite
grid.

In the $n$-copy discrimination problem, the observer receives
\begin{equation*}
    \rho_a^{(n)}(\theta)
    =
    \rho_a(\theta)^{\otimes n}
    \label{eq:ncopy_state}
\end{equation*}
and must infer the label $a$.
The nuisance parameter $\theta$ represents an external condition or apparatus
setting common to all candidate states.
It is not estimated by the observer.

\subsection{Measurements, success probability, and risk}

A $K$-outcome measurement on $\Hil^{\otimes n}$ is a POVM
\begin{equation*}
    M=\{M_a\}_{a=1}^K,\qquad
    M_a\succeq 0,\qquad
    \sum_{a=1}^K M_a=\id .
    \label{eq:povm}
\end{equation*}
When outcome $a$ is obtained, the observer decides that the prepared state was $\rho_a(\theta)$.
Let $q_a>0$ be the prior probability of label $a$, with $\sum_a q_a=1$.
For fixed $\theta$, the discrimination success probability is
\begin{equation*}
    P_{\rm succ}^{(n)}(M,\theta)
    =
    \sum_{a=1}^K q_a
    \Tr\!\left[
        \rho_a(\theta)^{\otimes n} M_a
    \right],
    \label{eq:success_def}
\end{equation*}
and the corresponding risk, i.e., the average error probability, is
\begin{equation}
    R_n(M,\theta)
    =
    1-P_{\rm succ}^{(n)}(M,\theta).
    \label{eq:risk_def}
\end{equation}

Because $\theta$ is unknown, a measurement should be evaluated by the entire profile
\begin{equation*}
    \theta\longmapsto R_n(M,\theta),
\end{equation*}
or equivalently by the success-probability profile.
The worst-case risk of a fixed measurement M is defined as
\begin{equation}
\overline R_n(M)
=
\sup_{\theta\in\Theta} R_n(M,\theta).
\label{eq:worst_case_risk}
\end{equation}
The objective of robust discrimination is to choose a measurement that minimizes this worst-case risk. Accordingly, we define the optimal robust risk by
\begin{equation}
R_n^{\rm rob}
=
\inf_M \overline R_n(M)
=
\inf_M \sup_{\theta\in\Theta} R_n(M,\theta).
\label{eq:minimax_risk_problem}
\end{equation}
A measurement attaining $R_n^{\rm rob}$
is called a minimax measurement. Equivalently, in terms of the success probability,
\begin{equation*}
\sup_M \inf_{\theta\in\Theta}
P_{\rm succ}^{(n)}(M,\theta)
=
1-R_n^{\rm rob}.
\end{equation*}

\subsection{Bayes risk over the nuisance parameter}

Let $w\in\cP(\Theta)$ be a probability distribution over the nuisance parameter.
The Bayes risk with respect to $w$ is
\begin{equation}
    r_n(M,w)
    =
    \int_\Theta R_n(M,\theta)\,w(d\theta).
    \label{eq:bayes_risk}
\end{equation}
This prior is not the prior over the candidate labels; the label prior is already given by $\{q_a\}$.
The prior $w$ describes uncertainty in the nuisance parameter.

Using Eq.~\eqref{eq:risk_def}, we obtain
\begin{align}
    r_n(M,w)
    &=
    1-
    \sum_{a=1}^K q_a
    \Tr\!\left[
        \bar\rho_a^{(n)}(w) M_a
    \right],
    \label{eq:bayes_risk_average_state}
\end{align}
where
\begin{equation*}
    \bar\rho_a^{(n)}(w)
    =
    \int_\Theta
    \rho_a(\theta)^{\otimes n}\,w(d\theta)
    \label{eq:averaged_state}
\end{equation*}
is the nuisance-averaged $n$-copy state.
Thus, for fixed $w$, minimizing the Bayes risk is equivalent to the ordinary Bayes-optimal discrimination problem for the averaged ensemble
\begin{equation*}
    \{\bar\rho_a^{(n)}(w),q_a\}_{a=1}^K .
\end{equation*}
This observation is the basis of our construction:
once an appropriate prior $w$ over the nuisance parameter is chosen, robust discrimination can be reduced to a standard minimum-error discrimination problem for averaged states.

\subsection{Finite-grid formulation}

The nuisance parameter $\theta$ takes values in the continuous parameter space $\Theta$. To numerically search for a least favorable prior over $\Theta$, however, the nuisance-parameter space must be represented in a finite-dimensional form. We therefore discretize $\Theta$ into a finite grid and represent a prior over the nuisance parameter by a probability vector on this grid. This discretization allows both the search for the LFP and the construction of the corresponding Bayes-optimal measurement to be carried out numerically. In what follows, we formulate the minimax discrimination problem on this finite grid.

We approximate $\Theta$ by a finite grid
\begin{equation*}
    \Theta_N=\{\theta_1,\ldots,\theta_N\}.
    \label{eq:theta_grid}
\end{equation*}
A prior over the grid is a probability vector
\begin{equation*}
    w=(w_1,\ldots,w_N)\in\Delta_N,
\end{equation*}
where
\begin{equation*}
    \Delta_N
    :=
    \left\{
        w\in\mathbb{R}^N
        \,\middle|\,
        w_i\ge 0,\quad
        \sum_{i=1}^N w_i=1
    \right\}
\end{equation*}
is the probability simplex.
Let
\begin{equation*}
    \sigma_{a,i}^{(n)}
    =
    \rho_a(\theta_i)^{\otimes n}.
    \label{eq:sigma_ai}
\end{equation*}
Then
\begin{equation}
    \bar\rho_a^{(n)}(w)
    =
    \sum_{i=1}^N w_i \sigma_{a,i}^{(n)}.
    \label{eq:discrete_average_state}
\end{equation}
The grid Bayes risk is
\begin{equation*}
    r_n(M,w)
    =
    \sum_{i=1}^N w_i R_i(M),
    \qquad
    R_i(M):=R_n(M,\theta_i).
    \label{eq:discrete_bayes_risk}
\end{equation*}
The finite-grid minimax risk is
\begin{equation*}
    V_{n,N}^{\rm risk}
    =
    \inf_M \max_{1\le i\le N} R_i(M),
    \label{eq:discrete_minimax_risk}
\end{equation*}
and the corresponding minimax success probability is
\begin{equation*}
    V_{n,N}^{\rm succ}
    =
    \sup_M \min_{1\le i\le N} P_i(M)
    =
    1-V_{n,N}^{\rm risk},
    \label{eq:discrete_minimax_success}
\end{equation*}
where $P_i(M)=P_{\rm succ}^{(n)}(M,\theta_i)$.
In the numerical study, the measurement is optimized on $\Theta_N$ and then checked on a finer evaluation grid.

%%%%%%%%%%%%%%%%%%%%%%%%%%%%%%%%%%%%%%%%%%%%%%%%%%%%%%%%%%%%%%%%%%%%%%%%%%
\section{Least favorable priors and minimax measurements}
\label{sec:lfp_minimax}
%%%%%%%%%%%%%%%%%%%%%%%%%%%%%%%%%%%%%%%%%%%%%%%%%%%%%%%%%%%%%%%%%%%%%%%%%%
We now establish the theoretical basis for constructing a robust measurement through a least favorable prior. We first define the LFP on the finite nuisance-parameter grid and then show, through a minimax theorem, how it is related to the original worst-case discrimination problem. We further derive the saddle-point and equalization properties that will be used to interpret the numerical solutions. Finally, we state the corresponding minimax relation for a continuous nuisance-parameter space.
% We now summarize the relation between least favorable priors and minimax measurements on a finite grid.
% The continuous case follows under standard compactness and continuity assumptions by the same minimax argument.

\subsection{Least favorable prior}

% \textcolor{red}{The minimax theorem establishes the optimal value and the correspondence between LFPs and minimax measurements, but it does not by itself describe the structure of an optimal pair. To interpret the resulting LFP and the success-probability profiles obtained numerically, we next examine the associated saddle-point conditions. In particular, these conditions reveal which nuisance-parameter values are active at the minimax solution and lead to an equalization property on the support of the LFP.
% }

Let $\cM$ be the set of all $K$-outcome POVMs on $\Hil^{\otimes n}$.
For $w\in\Delta_N$, define the Bayes-optimal risk
\begin{equation*}
    \beta_n(w)
    =
    \inf_{M\in\cM} r_n(M,w).
    \label{eq:bayes_value_beta}
\end{equation*}

\begin{definition}[Least favorable prior]
\label{def:lfp}
A prior $w^\star\in\Delta_N$ is called a least favorable prior if
\begin{equation*}
    \beta_n(w^\star)
    =
    \sup_{w\in\Delta_N} \beta_n(w).
    \label{eq:lfp_definition}
\end{equation*}
Equivalently, in the success-probability formulation, $w^\star$ minimizes the Bayes-optimal success probability over priors on the nuisance parameter.
\end{definition}

Thus, in the risk formulation, the LFP is the prior over $\Theta$ that makes the best achievable Bayes risk as large as possible.

\subsection{Minimax theorem}

The definition of an LFP alone does not yet show that optimizing against such a prior solves the original worst-case discrimination problem. We therefore establish the minimax relation between optimization over measurements and optimization over nuisance priors. This relation provides the theoretical justification for constructing a robust measurement as a Bayes-optimal measurement for an LFP.

We use Sion's minimax theorem \cite{sion1958}.
Applied to the compact convex set $\cM$, the simplex $\Delta_N$, and the affine function $r_n(M,w)$, it yields
\begin{equation}
    \inf_{M\in\cM}\sup_{w\in\Delta_N} r_n(M,w)
    =
    \sup_{w\in\Delta_N}\inf_{M\in\cM} r_n(M,w).
    \label{eq:sion_applied}
\end{equation}
For any fixed $M$,
\begin{equation}
    \sup_{w\in\Delta_N} r_n(M,w)
    =
    \max_{1\le i\le N} R_i(M),
    \label{eq:prior_max_equals_worst}
\end{equation}
because $r_n(M,w)$ is a convex combination of the finite set $\{R_i(M)\}_{i=1}^N$, and the maximum is attained by a point mass on an index achieving the largest risk.

\begin{theorem}[LFP--minimax correspondence]
\label{thm:lfp_minimax}
On the finite grid $\Theta_N$, the following statements hold.

\begin{enumerate}
\item The minimax identity
\begin{equation}
    \inf_{M\in\cM}\max_{1\le i\le N} R_i(M)
    =
    \sup_{w\in\Delta_N}\inf_{M\in\cM} r_n(M,w)
    \label{eq:minimax_lfp_equality}
\end{equation}
holds.

\item Let $w^\star$ be an LFP and let $M^\star$ be a minimax measurement.
Then $M^\star$ is Bayes optimal for $w^\star$:
\begin{equation}
    M^\star
    \in
    \arg\min_{M\in\cM} r_n(M,w^\star).
    \label{eq:minimax_is_bayes_lfp}
\end{equation}

\item Conversely, if $\widetilde M$ is Bayes optimal for $w^\star$ and satisfies
\begin{equation}
    \max_{1\le i\le N} R_i(\widetilde M)
    =
    r_n(\widetilde M,w^\star),
    \label{eq:bayes_to_minimax_condition}
\end{equation}
then $\widetilde M$ is a minimax measurement.
\end{enumerate}
\end{theorem}

\begin{proof}
Eq.~\eqref{eq:minimax_lfp_equality} follows from Eq.~\eqref{eq:sion_applied} and Eq.~\eqref{eq:prior_max_equals_worst}.
Let
\begin{equation*}
    V=\inf_{M\in\cM}\max_i R_i(M).
\end{equation*}
By the minimax identity and the definition of $w^\star$,
\begin{equation*}
    \inf_{M\in\cM} r_n(M,w^\star)=V.
\end{equation*}
If $M^\star$ is minimax, then $\max_i R_i(M^\star)=V$.
Since $r_n(M^\star,w^\star)$ is a convex combination of $\{R_i(M^\star)\}_i$, we have $r_n(M^\star,w^\star)\le V$.
On the other hand, because $V$ is the infimum of $r_n(M,w^\star)$ over all measurements, $r_n(M^\star,w^\star)\ge V$.
Thus equality holds and $M^\star$ is Bayes optimal for $w^\star$.
The converse follows immediately from Eq.~\eqref{eq:bayes_to_minimax_condition} and the equality $r_n(\widetilde M,w^\star)=V$.
\end{proof}

\subsection{Saddle-point condition and equalization}

The minimax theorem establishes the optimal value and the correspondence between LFPs and minimax measurements, but it does not by itself describe the structure of an optimal pair. To interpret the resulting LFP and the success-probability profiles obtained numerically, we next examine the associated saddle-point conditions. In particular, these conditions reveal which nuisance-parameter values are active at the minimax solution and lead to an equalization property on the support of the LFP.

\begin{proposition}[Saddle-point condition]
\label{prop:saddle_point}
If $M^\star$ is a minimax measurement and $w^\star$ is an LFP, then
\begin{equation}
    r_n(M^\star,w)
    \le
    r_n(M^\star,w^\star)
    \le
    r_n(M,w^\star)
    \label{eq:saddle_condition}
\end{equation}
for all $M\in\cM$ and $w\in\Delta_N$.
Moreover, $r_n(M^\star,w^\star)=V_{n,N}^{\rm risk}$.
\end{proposition}

\begin{proof}
Let
\[
V:=V^{\rm risk}_{n,N}
=
\inf_{M\in\mathcal M}\max_i R_i(M).
\]
Since $M^\star$ is minimax,
\[
\max_i R_i(M^\star)=V.
\]
Therefore, for any $w\in\Delta_N$,
\[
r_n(M^\star,w)
=
\sum_i w_i R_i(M^\star)
\le
\max_i R_i(M^\star)
=
V.
\]
In particular,
\[
r_n(M^\star,w^\star)\le V.
\]

On the other hand, since $w^\star$ is a least favorable prior,
Theorem~\ref{thm:lfp_minimax} gives
\[
\inf_{M\in\mathcal M}r_n(M,w^\star)=V.
\]
Hence
\[
r_n(M^\star,w^\star)
\ge
\inf_{M\in\mathcal M}r_n(M,w^\star)
=
V.
\]
Combining the two inequalities yields
\[
r_n(M^\star,w^\star)=V.
\]

Consequently, for any $w\in\Delta_N$,
\[
r_n(M^\star,w)
\le
V
=
r_n(M^\star,w^\star),
\]
and for any $M\in\mathcal M$,
\[
r_n(M^\star,w^\star)
=
V
\le
r_n(M,w^\star).
\]
Therefore,
\[
r_n(M^\star,w)
\le
r_n(M^\star,w^\star)
\le
r_n(M,w^\star),
\]
which proves the saddle-point condition.
\end{proof}

A minimax measurement $M^\star$ and an LFP $w^\star$ satisfying
Eq.~\eqref{eq:saddle_condition} are said to form a saddle point.

\begin{proposition}[Equalization on the LFP support]
\label{prop:equalizer_property}
Let $M^\star$ and $w^\star$ form a saddle point.
Then
\begin{equation*}
    R_i(M^\star)\le V_{n,N}^{\rm risk}
    \qquad
    \forall i,
\end{equation*}
and for every $i$ with $w_i^\star>0$,
\begin{equation}
    R_i(M^\star)=V_{n,N}^{\rm risk}.
    \label{eq:equalizer_support}
\end{equation}
Equivalently,
\begin{equation*}
    P_i(M^\star)=V_{n,N}^{\rm succ}
    \qquad
    (w_i^\star>0).
\end{equation*}
\end{proposition}

\begin{proof}
Since $M^\star$ is minimax, $\max_i R_i(M^\star)=V_{n,N}^{\rm risk}$, hence $R_i(M^\star)\le V_{n,N}^{\rm risk}$ for all $i$.
Because $r_n(M^\star,w^\star)=V_{n,N}^{\rm risk}$,
\begin{equation*}
    \sum_i w_i^\star
    \left[
        V_{n,N}^{\rm risk}-R_i(M^\star)
    \right]
    =
    0.
\end{equation*}
Each term is nonnegative.
Thus every index with positive weight must satisfy Eq.~\eqref{eq:equalizer_support}.
\end{proof}

This equalization property gives the LFP a direct physical meaning.
The support of $w^\star$ consists of nuisance-parameter values that are active worst cases for the minimax measurement.

\subsection{Continuous parameter spaces}

The finite-grid formulation is convenient for both the theoretical characterization above and the numerical construction developed in the next section. In the underlying physical problem, however, the nuisance parameter is generally continuous. It is therefore important to clarify that the finite-grid formulation is not a separate decision problem, but a discretization of an analogous minimax problem on the continuous parameter space. We briefly state this continuous counterpart below.

If $\Theta$ is compact and $\theta\mapsto\rho_a(\theta)$ is continuous for each $a$, then $\theta\mapsto R_n(M,\theta)$ is continuous for every fixed $M$.
Let $\cP(\Theta)$ denote the probability measures on $\Theta$ with the weak topology.
Under the same convexity and compactness assumptions, Sion's theorem gives
\begin{equation}
    \inf_M\sup_{\theta\in\Theta} R_n(M,\theta)
    =
    \sup_{w\in\cP(\Theta)}
    \inf_M
    \int_\Theta R_n(M,\theta)\,w(d\theta).
    \label{eq:continuous_minimax_lfp}
\end{equation}
Thus the finite-grid construction used below should be viewed as a computable discretization of this continuous minimax problem.

\section{Numerical results}
\label{sec:numerical_results}
%%%%%%%%%%%%%%%%%%%%%%%%%%%%%%%%%%%%%%%%%%%%%%%%%%%%%%%%%%%%%%%%%%%%%%%%%%

The preceding sections establish the theoretical basis of the LFP-based minimax approach. The SDP implementation used in the numerical calculations is detailed in Appendix \ref{sec:numerical_construction}. It remains to examine how the resulting measurement performs in explicit quantum-state discrimination problems and how much robustness is gained relative to natural benchmark strategies. In this section, we therefore evaluate the proposed method numerically by comparing its success-probability profile and worst-case performance with those of a reference-point optimal measurement and a uniform-prior Bayes measurement. We first introduce the measurement strategies used for comparison and then consider two examples of increasing complexity: a binary nonorthogonal qubit model illustrating the basic effect of the LFP-based construction, followed by a three-state nonorthogonal qutrit model demonstrating its applicability beyond binary discrimination.

\subsection{Measurement strategies}
\label{sec:measurement_strategies}
%%%%%%%%%%%%%%%%%%%%%%%%%%%%%%%%%%%%%%%%%%%%%%%%%%%%%%%%%%%%%%%%%%%%%%%%%%

We compare four measurement strategies.
All of them are fixed measurements on $\Hil^{\otimes n}$ and are evaluated as functions of the unknown parameter $\theta$.

\subsubsection{Reference-point optimal measurement}
\label{subsec:reference_measurement}

Let $\theta_{\rm ref}\in\Theta$ be a nominal reference point.
The reference-point optimal measurement is
\begin{equation*}
    M_{\rm ref}^{(n)}
    \in
    \arg\max_M
    P_{\rm succ}^{(n)}(M,\theta_{\rm ref}).
    \label{eq:reference_measurement_success}
\end{equation*}
This measurement is optimal if the true parameter equals $\theta_{\rm ref}$.
However, if $\theta$ differs substantially from $\theta_{\rm ref}$, the measurement may be poorly aligned with the actual perturbed ensemble, and the success probability can deteriorate.
It is therefore a useful baseline for assessing robustness.

\subsubsection{Uniform-prior Bayes measurement}
\label{subsec:uniform_measurement}

The uniform-prior Bayes measurement uses the uniform distribution on the optimization grid,
\begin{equation*}
    w_{{\rm unif},i}=\frac{1}{N}.
\end{equation*}
It is defined by
\begin{equation*}
    M_{\rm unif}^{(n)}
    \in
    \arg\min_M r_n(M,w_{\rm unif}).
    \label{eq:uniform_measurement}
\end{equation*}
This strategy accounts for the whole parameter grid in an average sense.
It is often more stable than the reference-point measurement, but it does not specifically emphasize the parameter values that dominate the worst-case performance.

\subsubsection{LFP-based robust measurement (global collective strategy)}
\label{subsec:global_lfp_measurement}

The global LFP-based measurement is the Bayes-optimal collective
measurement for a least favorable prior,
\begin{equation*}
    M_{\rm LFP}^{\rm global}
    \in
    \arg\min_M r_n(M,w^\star).
\end{equation*}
The global LFP-based measurement $M_{\rm LFP}^{\rm global}$
and the corresponding least favorable prior $w^\star$
are obtained as the primal and dual optimal solutions, respectively, of the minimax SDP described in Appendix~\ref{sec:numerical_construction}. 
The measurement is therefore minimax on the finite optimization grid by construction.
Here the optimization is performed over arbitrary POVMs on
$\Hil^{\otimes n}$.
On the finite grid, when the saddle-point condition is satisfied,
this measurement achieves the minimax value and equalizes the
success probability on the support of the LFP.

\subsubsection{Local-product comparison strategy}
\label{subsec:local_lfp_measurement}

To assess the advantage of collective measurements, we also consider
a local-product strategy constructed from the same least favorable
prior $w^\star$.
We first determine a single-copy Bayes-optimal POVM
$F^\star=\{F_b^\star\}$ for the nuisance-averaged single-copy ensemble,
\begin{equation*}
    F^\star
    \in
    \arg\max_{\{F_b\}}
    \sum_{a=1}^{K} q_a
    \Tr\!\left[
        \bar{\rho}^{(1)}_a(w^\star)F_a
    \right],
\end{equation*}
where
\begin{equation*}
    \bar{\rho}^{(1)}_a(w^\star)
    =
    \sum_{i=1}^{N}w_i^\star \rho_a(\theta_i).
\end{equation*}

The POVM $F^\star$ is then applied independently to each of the
$n$ copies.
For the resulting outcome string
$\bm b=(b_1,\ldots,b_n)$, the score assigned to label $a$ is
\begin{equation*}
    A_a(\bm b;w^\star,F^\star)
    =
    q_a
    \sum_{i=1}^{N} w_i^\star
    \prod_{\ell=1}^{n}
    \Tr\!\left[
        \rho_a(\theta_i)F_{b_\ell}^\star
    \right].
\end{equation*}
The final decision is made according to the Bayes rule
\begin{equation*}
    \delta(\bm b)
    \in
    \arg\max_a
    A_a(\bm b;w^\star,F^\star).
\end{equation*}

This procedure defines the local-product measurement
$M_{\rm LFP}^{\rm local}$ used in the numerical comparisons.
Unlike $M_{\rm LFP}^{\rm global}$, it does not allow a collective
quantum measurement across the $n$ copies; only the final classical
decision uses the complete outcome string.
The local-product strategy is included as a benchmark and is not
intended to represent an optimization over all adaptive local
measurement protocols.

\subsection{Binary nonorthogonal qubit model}
\label{subsec:binary_nonorthogonal_qubit_model}

We first consider a simple binary qubit model in order to illustrate the basic role of the LFP-based minimax measurement.
Let ${|0\rangle,|1\rangle}$ be the computational basis of $\mathbb{C}^2$, and define two nonorthogonal pure states by
\begin{equation*}
|\psi_0\rangle = |0\rangle,
\qquad
|\psi_1\rangle = \cos\alpha|0\rangle+\sin\alpha|1\rangle .
\end{equation*}
The corresponding density operators are
\begin{equation*}
\rho_a = |\psi_a\rangle\langle\psi_a|,
\qquad a=0,1 .
\end{equation*}
Both candidate states are affected by the same unknown rotation about the $y$ axis,
\begin{equation*}
\rho_a(\theta)
=
U(\theta)\rho_a U(\theta)^\dagger,
\qquad
U(\theta)
=
\exp\left(
-\frac{i\theta\sigma_y}{2}
\right),
\end{equation*}
where $\theta\in[-\theta_{\max},\theta_{\max}]$.

In this numerical experiment, we set $\alpha=\pi/4$, $\theta_{\max}=\pi/3$ and  equal label priors, $q_0 = q_1 = 1/2$.

The overlap between the two unperturbed states is
\begin{equation*}
|\langle\psi_0|\psi_1\rangle|^2
=
\cos^2\alpha .
\end{equation*}
For $\alpha=\pi/4$, this overlap is equal to $1/2$.
Since the same unitary perturbation is applied to both candidate states, the overlap
\begin{equation*}
\Tr[\rho_0(\theta)\rho_1(\theta)]
=
\Tr[\rho_0\rho_1]
\end{equation*}
is independent of $\theta$.
Therefore, if $\theta$ were known, the intrinsic difficulty of the binary discrimination problem would be the same for all $\theta$.
The robust problem is nevertheless nontrivial because the measurement must be fixed before the value of the nuisance parameter is known.
A measurement optimized at one reference value of $\theta$ can be poorly aligned with the rotated states at another value of $\theta$.

We consider $n=3$ copies and compare the four measurement strategies
introduced in Sec.~\ref{sec:measurement_strategies}.
The reference-point measurement is optimized at
$\theta_{\rm ref}=0$.
The optimization grid consists of $N=41$ points, and the resulting
measurements are evaluated on a finer grid with
$N_{\rm eval}=401$ points.

Fig.~\ref{fig:qubit_nonorthogonal_profiles} shows the success-probability profiles of the four measurements.
The reference-point optimal measurement performs well near $\theta=0$, but its performance rapidly deteriorates toward the edge of the parameter interval.
This illustrates that a nominally optimized measurement can be highly sensitive to an unknown common unitary perturbation, even though the intrinsic overlap between the candidate states is unchanged.
The uniform-prior Bayes measurement improves robustness by taking the entire parameter range into account.
However, because it weights all parameter values uniformly, it does not necessarily focus on the parameter values that determine the minimax performance.
In contrast, the global LFP-based minimax measurement raises the lowest part of the success-probability profile and gives the largest worst-case success probability among the four strategies.

\begin{figure}[t]
\centering
\includegraphics[width=0.85\linewidth]{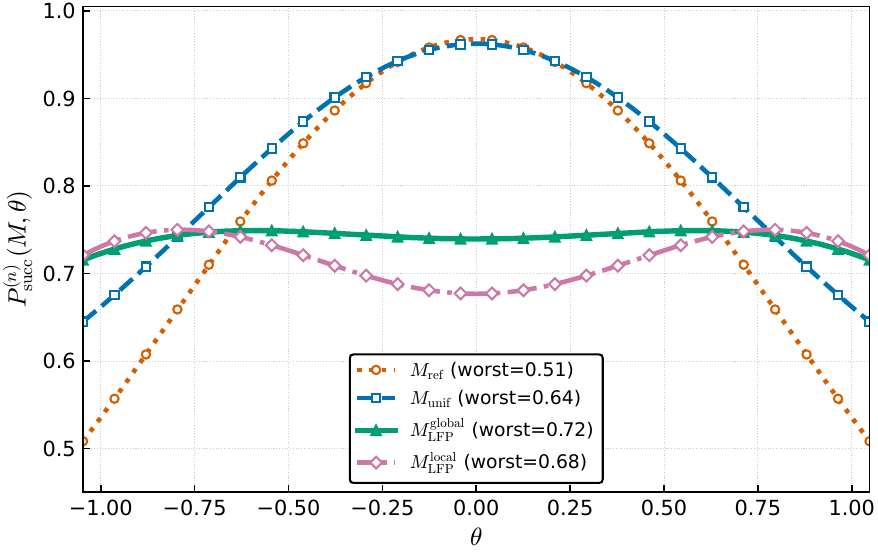}
\caption{
Success-probability profiles for the binary nonorthogonal qubit model.
}
\label{fig:qubit_nonorthogonal_profiles}
\end{figure}

Table~\ref{tab:qubit_worst_case_success} summarizes the worst-case success probabilities.
The reference-point optimal measurement has worst-case success probability $0.51$ on the evaluation grid, showing that a measurement optimized at a single nominal value can become almost uninformative at the edge of the perturbation range.
The uniform-prior Bayes measurement improves the worst-case value
to $0.64$.
The local-product LFP-based strategy further improves the worst-case
success probability to $0.68$, while the global collective LFP-based
measurement achieves $0.72$.
Thus, in this simple nonorthogonal qubit benchmark, the global
LFP-based construction improves the worst-case success probability
by approximately $0.21$ compared with $M_{\rm ref}$ and by
approximately $0.08$ compared with $M_{\rm unif}$.

\begin{table}[t]
\centering
\caption{
Worst-case success probabilities for the binary nonorthogonal qubit model.
}
\label{tab:qubit_worst_case_success}
\begin{tabular}{lccc}
\toprule
Measurement & Optimization grid & Evaluation grid & Worst $\theta$ \\
\midrule
$M_{\rm ref}$                    & 0.51 & 0.51 & $\pm\theta_{\max}$ \\
$M_{\rm unif}$                   & 0.64 & 0.64 & $\pm\theta_{\max}$ \\
$M_{\rm LFP}^{\rm local}$        & 0.68 & 0.68 & $0$ \\
$M_{\rm LFP}^{\rm global}$       & 0.72 & 0.72 & $\pm\theta_{\max}$ \\
\bottomrule
\end{tabular}
\end{table}

This example demonstrates the central purpose of the LFP-based minimax approach.
Even though the pairwise overlap of the two candidate states is invariant under the common unitary perturbation, the performance of a fixed measurement depends strongly on the unknown value of $\theta$.
The LFP-based measurement explicitly targets the nuisance-parameter values that limit the worst-case performance and therefore provides a constructive way to design a robust measurement.
In the following qutrit example, we show that the same idea remains effective in a genuinely three-state, nonorthogonal, and asymmetric model.

\subsection{Three-state qutrit model}
\label{subsec:qutrit_three_state_model}

We next apply the proposed method to a three-state discrimination problem on $\Hil=\mathbb{C}^3$.
Let ${|0\rangle,|1\rangle,|2\rangle}$ be an orthonormal basis and consider three pure states
\begin{equation*}
\rho_a=|\psi_a\rangle\langle\psi_a|,
\qquad a=0,1,2 .
\end{equation*}
The states are defined by
\begin{align}
|\psi_0\rangle
&=
\sqrt{1-r}|0\rangle
+
\sqrt r|1\rangle ,
\label{eq:qutrit_psi0_specific}\\
|\psi_1\rangle
&=
\sqrt{1-r}|0\rangle
+
\sqrt r\left(
\cos\eta_{\rm asym}|1\rangle
+
\sin\eta_{\rm asym}|2\rangle
\right),
\label{eq:qutrit_psi1_specific}\\
|\psi_2\rangle
&=
\sqrt{1-r}|0\rangle
+
\sqrt r\left(
\cos\eta_{\rm asym}|1\rangle
-
\sin\eta_{\rm asym}|2\rangle
\right).
\label{eq:qutrit_psi2_specific}
\end{align}
In the numerical experiment, we set $r=0.35$ and $\eta_{\rm asym}=0.45\pi$.

Since the states are pure, their pairwise overlaps are given by
\begin{equation*}
\Tr[\rho_a\rho_b]
=
|\langle\psi_a|\psi_b\rangle|^2 .
\end{equation*}
For the present model, this gives
\begin{align}
\Tr[\rho_0\rho_1]
&=
\Tr[\rho_0\rho_2]
=
\left[
(1-r)+r\cos\eta_{\rm asym}
\right]^2,
\label{eq:qutrit_overlap_01_02}\\
\Tr[\rho_1\rho_2]
&=
\left[
(1-r)+r\cos(2\eta_{\rm asym})
\right]^2 .
\label{eq:qutrit_overlap_12}
\end{align}
With $r=0.35$ and $\eta_{\rm asym}=0.45\pi$, these values are
$\Tr[\rho_0\rho_1]=\Tr[\rho_0\rho_2]\simeq 0.5$ and
$\Tr[\rho_1\rho_2]\simeq 0.1$.
Thus $\rho_0$ has a relatively large overlap with both $\rho_1$ and $\rho_2$, while $\rho_1$ and $\rho_2$ are better separated.
This asymmetric nonorthogonal structure makes the model a useful test bed beyond symmetric binary or qubit examples.

The common unitary perturbation is defined by
\begin{equation*}
U(\theta)=\exp[-i\phi(\theta)H],
\qquad
\theta\in[-\beta,\beta],
\end{equation*}
where $\beta=1.0$, $\phi_{\max}=2.0$, and
\begin{equation*}
\phi(\theta)
=
\phi_{\max}
\left(\frac{\theta}{\beta}\right)^3 .
\label{eq:qutrit_phi_specific}
\end{equation*}
The Hamiltonian is
\begin{equation*}
H
=
\frac{H_0}{||H_0||_{\rm op}},
\qquad
H_0 =
\begin{pmatrix}
0 & 1 & 0.7\\
1 & 0 & 1\\
0.7 & 1 & 0
\end{pmatrix}.
\label{eq:qutrit_hamiltonian_specific}
\end{equation*}
where $||H_0||_{\rm op}$ denotes the operator norm of $H_0$.
Since $H_0$ is Hermitian, this norm is equal to the largest absolute eigenvalue,
$||H_0||_{\rm op}=\max_j|\lambda_j(H_0)|$.
The normalization factor is therefore chosen so that the largest absolute eigenvalue of $H$ is one.

This Hamiltonian mixes all three levels nontrivially and is not reducible to a simple two-level rotation.
The use of the nonlinear phase function $\phi(\theta)$ further makes the parameter dependence nonuniform over the nuisance-parameter range.

We consider $n=3$ copies and use the uniform label prior $q_a=1/3$.
For a POVM $M=\{M_0,M_1,M_2\}$ on $\Hil^{\otimes n}$, the success probability is
\begin{equation*}
P_{\rm succ}^{(n)}(M,\theta)
=
\frac{1}{3}
\sum_{a=0}^2
\Tr \left[
\rho_a(\theta)^{\otimes n}M_a
\right].
\label{eq:qutrit_success_probability_specific}
\end{equation*}
The optimization grid consists of $N=31$ points, while the success-probability profiles are evaluated on a finer grid with $N_{\rm eval}=1001$ points.
The finite-grid minimax problem is solved directly as a semidefinite program.
The SDP solver is SCS with tolerance $10^{-5}$ and maximum iteration number $5000$.

We compare the four measurement strategies introduced in
Sec.~\ref{sec:measurement_strategies}.
The reference-point measurement is optimized at
$\theta_{\rm ref}=0$.
Fig.~\ref{fig:qutrit_success_profiles} compares the success-probability profiles of the four measurements.

The reference-point optimal measurement $M_{\rm ref}$ achieves a high Bayes value of $0.96$ at $\theta_{\rm ref}=0$.
However, its worst-case success probability on the evaluation grid drops to $0.47$.
This strong degradation shows that a measurement optimized at a single reference point can be highly sensitive to a common unitary perturbation.
Near $\theta=0$, the measurement is well matched to the nominal ensemble, whereas near the endpoints the perturbed states are substantially rotated relative to the fixed measurement.

The uniform-prior Bayes measurement $M_{\rm unif}$ improves the worst-case success probability to $0.79$ on the evaluation grid.
This improvement is expected because the measurement is optimized for the average ensemble over the entire parameter grid.
Nevertheless, the uniform prior treats easy and hard parameter regions equally.
As a result, it does not fully target the parameter values that determine the minimax performance.

The local-product LFP-based strategy achieves a worst-case success
probability of $0.83$, while the global collective LFP-based
measurement further improves it to $0.86$.
Thus, the global LFP-based measurement improves the worst-case
success probability by approximately $0.39$ compared with
$M_{\rm ref}$ and by approximately $0.07$ compared with
$M_{\rm unif}$.
The agreement between the optimization and evaluation grids
indicates that, for this discretization, no significant deterioration
occurs between grid points.

The success-probability profile of $M_{\rm LFP}$ is visibly flatter than those of the other strategies.
This behavior reflects the equalization property of minimax solutions.
The LFP assigns weight to parameter values that limit the worst-case performance, and the corresponding Bayes-optimal measurement raises the valleys of the success-probability profile rather than maximizing performance at a single reference point.
Therefore, this qutrit example demonstrates that the LFP-based construction gives a robust measurement even in a genuinely three-state, nonorthogonal, and asymmetric model.

\begin{figure}[t]
\centering
\includegraphics[width=0.85\linewidth]{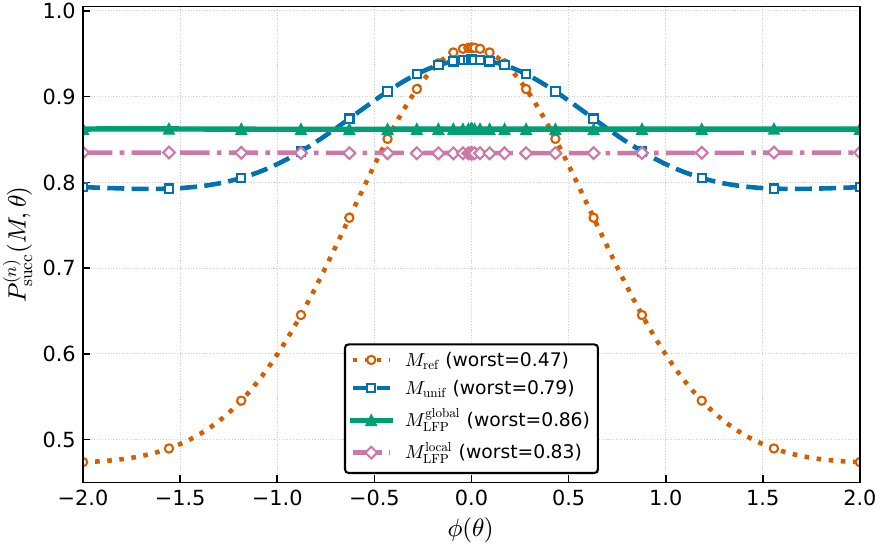}
\caption{
Success-probability profiles for the three-state qutrit model.
}
\label{fig:qutrit_success_profiles}
\end{figure}

Fig.~\ref{fig:qutrit_lfp} shows the least favorable prior
$w^\star$ obtained on the optimization grid.
The LFP is strongly nonuniform and assigns appreciable weight not only
near the boundaries of the nuisance-parameter interval but also to
several interior grid points.
This structure shows that the least favorable nuisance distribution is
not determined solely by the magnitude of the unitary perturbation.

Together with Fig.~\ref{fig:qutrit_success_profiles}, this result also
illustrates the equalizing behavior of the minimax construction.
Although the LFP itself is highly nonuniform, the corresponding global
LFP-based measurement produces an almost flat success-probability
profile over the nuisance-parameter range.
Thus, the LFP provides a nonuniform nuisance distribution whose
Bayes-optimal collective measurement realizes robust, nearly equalized
discrimination performance.

\begin{figure}[t]
\centering
\includegraphics[width=0.85\linewidth]
{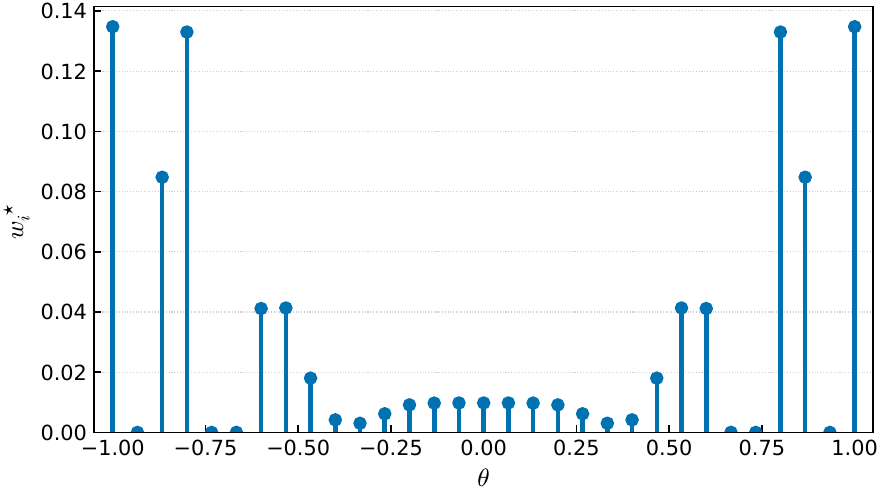}
\caption{
Least favorable prior $w^\star$ on the nuisance-parameter grid for the
three-state qutrit model.
}
\label{fig:qutrit_lfp}
\end{figure}

Table~\ref{tab:qutrit_worst_case_success} summarizes the worst-case success probabilities and the corresponding nuisance-parameter values for the four measurements. The global LFP-based measurement achieves the highest worst-case success probability, $0.86$, followed by the local-product LFP-based measurement at $0.83$, the uniform-prior Bayes measurement at $0.79$, and the reference-point measurement at $0.47$. The worst-case location also depends strongly on the measurement: the reference-point measurement is limited by the endpoints of the parameter interval, whereas the worst cases of the LFP-based measurements occur at different, including interior, parameter values. This observation further illustrates that the difficulty of robust discrimination is determined not only by the magnitude of the perturbation but also by its interplay with the chosen measurement. The optimization-grid and evaluation-grid values agree to the reported precision for all four strategies.

% Table~\ref{tab:qutrit_worst_case_success} summarizes the worst-case success probabilities.
% The reported gap is the difference between the worst-case success probability on the optimization grid and that on the finer evaluation grid.

\begin{table}[t]
\centering
\caption{
Worst-case success probabilities for the three-state qutrit model.
}
\label{tab:qutrit_worst_case_success}
\begin{tabular}{lccc}
\toprule
Measurement & Optimization grid & Evaluation grid & Worst $\theta$ \\
\midrule
$M_{\rm ref}$
& 0.47 & 0.47 & $\pm \beta$ \\
$M_{\rm unif}$
& 0.79 & 0.79 & $\pm 0.94$ \\
$M_{\rm LFP}^{\rm local}$
& 0.83& 0.83 & $  0$ \\
$M_{\rm LFP}^{\rm global}$
& 0.86 & 0.86 & $ \pm 0.65$ \\
\bottomrule
\end{tabular}
\end{table}
%%%%%%%%%%%%%%%%%%%%%%%%%%%%%%%%%%%%%%%%%%%%%%%%%%%%%%%%%%%%%%%%%%%%%%%%%%
\section{Discussion}
\label{sec:discussion}
%%%%%%%%%%%%%%%%%%%%%%%%%%%%%%%%%%%%%%%%%%%%%%%%%%%%%%%%%%%%%%%%%%%%%%%%%%

\subsection{Physical meaning of the LFP}

In the present setting, $\theta$ represents an unknown external condition, such as a reference-frame mismatch, phase drift, polarization rotation, calibration error, or slowly varying unitary drift.
The observer is not trying to estimate this parameter.
The goal is to identify the label $a$ despite not knowing $\theta$.

The LFP can therefore be interpreted as the distribution over unknown conditions that makes the discrimination task most difficult.
% In the risk formulation,
% \begin{equation*}
%     w^\star\in\arg\max_w \inf_M r_n(M,w).
% \end{equation*}
This is the distribution chosen by a hypothetical adversary that tries to maximize the Bayes risk while the observer chooses the best measurement for that distribution~\cite{ferguson1967}.
Positive mass in the LFP indicates nuisance-parameter values that are active worst cases for the minimax measurement.

If the LFP concentrates at the endpoints of a parameter interval, the largest unitary mismatch dominates the robust performance.
If the LFP also assigns weight to interior points, then the difficulty is not determined solely by perturbation strength.
It also reflects the geometry of the candidate states, the direction of the unitary action, and competition among the POVM elements.
Thus the LFP is not merely an auxiliary object in a minimax proof; it is a diagnostic tool for understanding the geometry of robust quantum discrimination.

\subsection{Flattening of success-probability profiles}

The most visible effect of the LFP-based measurement is the flattening of the success-probability profile
\begin{equation*}
    \theta\longmapsto P_{\rm succ}^{(n)}(M,\theta).
\end{equation*}
The reference-point measurement can achieve excellent performance at $\theta_{\rm ref}$, but it may fail away from that point.
The uniform-prior measurement generally improves average stability but may still underweight rare or localized difficult parameter values.
The LFP-based measurement, by contrast, directly targets the worst-case risk.

On a finite grid, the equalization property implies
\begin{equation*}
    P_{\rm succ}^{(n)}(M_{\rm LFP}^{\rm global},\theta_i)
    =
    V_{n,N}^{\rm succ}
    \qquad
    (w_i^\star>0),
\end{equation*}
whenever the saddle-point condition holds.
The flattening observed numerically is therefore a consequence of the minimax--LFP duality rather than a numerical coincidence.

\subsection{Robustness versus reference-point performance}

Robustness generally comes with a tradeoff.
The reference-point measurement is optimal at $\theta_{\rm ref}$ and can be preferable when the experimental condition is known to be close to that point.
The LFP-based measurement may sacrifice some nominal performance in order to increase
\begin{equation*}
    \inf_{\theta\in\Theta} P_{\rm succ}^{(n)}(M,\theta).
\end{equation*}
This is the natural tradeoff in minimax decision theory.

Thus the LFP measurement should not be interpreted as universally best under every criterion.
It is the appropriate choice when the objective is to control the worst-case risk, especially when the nuisance parameter is poorly calibrated, adversarial, or not described by a reliable physical prior.

\subsection{Scalability and limitations}

The framework applies to any finite number of hypotheses, finite-dimensional Hilbert spaces, and general nuisance-parameter models.
Nevertheless, the computational cost grows rapidly.
For a single-copy dimension $d$, the $n$-copy Hilbert space has dimension $d^n$.
The POVM elements and the dual variable in the SDP are therefore $d^n\times d^n$ matrices.
This limits direct collective-measurement optimization at large $n$.

The grid size $N$ also affects the computation.
For multiparameter nuisance models, a naive grid suffers from the curse of dimensionality.
Adaptive grids, sampling-based methods, cutting-plane methods, or exchange methods may be needed for larger parameter spaces~\cite{OustryCerulli2025,Jungen2022}.
Developing scalable algorithms for high-dimensional nuisance-parameter spaces is an important direction for future work.

Another limitation is that the finite-grid SDP solves the discretized minimax problem exactly, not the continuous problem directly.
One should therefore evaluate the resulting measurement on a finer grid and check stability under changes of $N$.
In the present qutrit example, the optimization-grid and evaluation-grid worst-case values agree to the reported precision, supporting the robustness of the discretization.

Finally, although the present work focuses on common unitary perturbations, the same minimax--LFP framework can be extended to parameter-dependent channels
\begin{equation*}
    \rho_a\longmapsto \mathcal{E}_\theta(\rho_a).
\end{equation*}
More general quantum-process discrimination problems can likewise be
formulated as convex optimization problems, including minimax variants
\cite{nakahira2021process}.
This would include loss, depolarization, amplitude damping, and other nonunitary noise mechanisms.
In such models, the spectra and distinguishability of the states may themselves vary with $\theta$, and the structure of the LFP may be more complex.

%%%%%%%%%%%%%%%%%%%%%%%%%%%%%%%%%%%%%%%%%%%%%%%%%%%%%%%%%%%%%%%%%%%%%%%%%%
\section{Conclusion}
\label{sec:conclusion}
%%%%%%%%%%%%%%%%%%%%%%%%%%%%%%%%%%%%%%%%%%%%%%%%%%%%%%%%%%%%%%%%%%%%%%%%%%

We formulated multi-hypothesis quantum-state discrimination under an unknown common unitary perturbation as a minimax decision problem with a nuisance parameter.
The label $a$ is the object to be identified, whereas $\theta$ is an unknown external condition that affects the performance of a fixed measurement.

The theoretical basis of our construction is the correspondence between least favorable priors and minimax measurements.
For a finite grid of nuisance-parameter values, the LFP is a prior over $\theta$ that maximizes the Bayes-optimal risk, and a Bayes measurement for this prior gives a minimax measurement when the saddle-point conditions are satisfied.
We showed that the finite-grid minimax measurement and its least favorable prior arise as the primal and dual solutions, respectively, of a pair of semidefinite programs.

The qutrit numerical example demonstrates that the method is not restricted to binary or qubit discrimination.
For three nonorthogonal qutrit states subject to an unknown common unitary perturbation, the LFP-based measurement improves the worst-case success probability and flattens the success-probability profile.
The LFP identifies the nuisance-parameter regions that actively constrain robust performance.

These results show that least favorable priors are not only abstract objects in minimax quantum decision theory.
They provide a constructive and interpretable tool for designing robust quantum measurements under experimental uncertainty.
Future work includes continuous-parameter algorithms, multiparameter nuisance models, nonunitary parameter-dependent channels, and symmetry-based reductions for collective measurements.

\begin{acknowledgments}
    This study was supported by JSPS KAKENHI Grant Numbers JP23H01432, JP23K11006 and JST BOOST, Japan Grant Number JPMJBS2402.
\end{acknowledgments}
%%%%%%%%%%%%%%%%%%%%%%%%%%%%%%%%%%%%%%%%%%%%%%%%%%%%%%%%%%%%%%%%%%%%%%%%%%
\appendix
%%%%%%%%%%%%%%%%%%%%%%%%%%%%%%%%%%%%%%%%%%%%%%%%%%%%%%%%%%%%%%%%%%%%%%%%%%

\section{Numerical construction}
\label{sec:numerical_construction}

The preceding section establishes the theoretical connection between
least favorable priors and minimax measurements.
To use this characterization for numerical measurement design, we need
a concrete procedure for computing both quantities.
In this section, we formulate the finite-grid problem introduced in
Sec.~\ref{sec:formulation} in terms of semidefinite programming.
We first recall the SDP formulation of the Bayes-optimal measurement
for a fixed nuisance prior.
We then formulate the finite-grid minimax problem directly as a primal
SDP and show that its dual problem yields a least favorable prior.
This primal--dual formulation directly provides a minimax measurement
and the corresponding LFP, avoiding any ambiguity associated with
nonunique Bayes-optimal measurements.

%%%%%%%%%%%%%%%%%%%%%%%%%%%%%%%%%%%%%%%%%%%%%%%%%%%%%%%%%%%%%%%%%%%%%%%%%%

\subsection{Bayes-optimal POVM as an SDP}

For a fixed nuisance prior, the discrimination problem reduces to an
ordinary minimum-error discrimination problem for the
nuisance-averaged ensemble.
Minimum-error quantum-state discrimination admits a
semidefinite-programming formulation, which provides both the optimal
POVM and a dual representation of the optimal success probability
\cite{eldar2003,nakahira2015}.

For fixed $w\in\Delta_N$, define the nuisance-averaged states by
Eq.~\eqref{eq:discrete_average_state}.
The Bayes-optimal success probability is
\begin{align}
    S_n(w)
    =
    \sup_{\{M_a\}}
    \quad&
    \sum_{a=1}^K q_a
    \Tr\!\left[
        \bar\rho_a^{(n)}(w)M_a
    \right]
    \label{eq:bayes_primal_sdp}\\
    \mathrm{s.t.}
    \quad&
    M_a\succeq0,\qquad a=1,\ldots,K,\notag\\
    &
    \sum_{a=1}^K M_a=\id .\notag
\end{align}
The corresponding Bayes-optimal risk is
\begin{equation*}
    \beta_n(w)=1-S_n(w).
\end{equation*}

The dual SDP is
\begin{align}
    S_n(w)
    =
    \inf_{Z=Z^\dagger}
    \quad&
    \Tr Z
    \label{eq:bayes_dual_sdp}\\
    \mathrm{s.t.}
    \quad&
    Z\succeq q_a\bar\rho_a^{(n)}(w),
    \qquad a=1,\ldots,K .\notag
\end{align}
Slater's condition holds, for example by choosing
$M_a=\id/K$, and hence strong duality holds
\cite{boyd2004}.
The optimal primal and dual variables satisfy the
Holevo--Helstrom conditions
\begin{equation}
    Z^\star(w)\succeq
    q_a\bar\rho_a^{(n)}(w),
    \qquad a=1,\ldots,K,
    \label{eq:hh_condition_1}
\end{equation}
and
\begin{equation}
    \left[
        Z^\star(w)-q_a\bar\rho_a^{(n)}(w)
    \right]
    M_a^\star(w)
    =
    0,
    \qquad a=1,\ldots,K .
    \label{eq:hh_condition_2}
\end{equation}

%%%%%%%%%%%%%%%%%%%%%%%%%%%%%%%%%%%%%%%%%%%%%%%%%%%%%%%%%%%%%%%%%%%%%%%%%%

\subsection{Direct SDP for the finite-grid minimax measurement}

We next formulate the finite-grid minimax problem directly.
For each grid point $\theta_i$, define
\begin{equation*}
    P_i(M)
    :=
    \sum_{a=1}^K
    q_a
    \Tr\!\left[
        \sigma_{a,i}^{(n)}M_a
    \right].
\end{equation*}
As defined in Sec.~\ref{sec:formulation}, the finite-grid minimax
success probability is
\begin{equation*}
    V_{n,N}^{\rm succ}
    =
    \sup_M
    \min_{1\le i\le N} P_i(M).
\end{equation*}

Introducing an auxiliary variable $t$ that lower-bounds the success
probability at every grid point, this problem can be written as the SDP
\begin{align}
    V_{n,N}^{\rm succ}
    =
    \sup_{\{M_a\},t}
    \quad&
    t
    \label{eq:minimax_primal_sdp}\\
    \mathrm{s.t.}
    \quad&
    \sum_{a=1}^K
    q_a
    \Tr\!\left[
        \sigma_{a,i}^{(n)}M_a
    \right]
    \ge t,
    \qquad i=1,\ldots,N,
    \notag\\
    &
    M_a\succeq0,
    \qquad a=1,\ldots,K,
    \notag\\
    &
    \sum_{a=1}^K M_a=\id .
    \notag
\end{align}

Let $M^\star=\{M_a^\star\}_{a=1}^K$ and $t^\star$ be an
optimal solution of Eq.~\eqref{eq:minimax_primal_sdp}.
By construction,
\begin{equation}
    t^\star
    =
    V_{n,N}^{\rm succ},
    \qquad
    P_i(M^\star)\ge V_{n,N}^{\rm succ}
    \quad
    \text{for all }i .
    \label{eq:minimax_primal_property}
\end{equation}
Therefore, $M^\star$ is a minimax measurement on the finite
optimization grid.

This direct formulation is useful because the measurement obtained
from Eq.~\eqref{eq:minimax_primal_sdp} is minimax by construction.
In particular, it does not require selecting a particular
Bayes-optimal measurement from a possibly nonunique set of solutions.

%%%%%%%%%%%%%%%%%%%%%%%%%%%%%%%%%%%%%%%%%%%%%%%%%%%%%%%%%%%%%%%%%%%%%%%%%%

\subsection{Dual SDP and the least favorable prior}

We now derive the dual of Eq.~\eqref{eq:minimax_primal_sdp}.
Introduce a nonnegative dual variable $w_i$ for each constraint
$P_i(M)\ge t$, and a Hermitian matrix $Z$ for the POVM completeness
constraint.
The Lagrangian is
\begin{align}
    \mathcal{L}
    =
    t
    +
    \sum_{i=1}^N
    w_i
    \left[
        P_i(M)-t
    \right]
    +
    \Tr\!\left[
        Z
        \left(
            \id-\sum_{a=1}^K M_a
        \right)
    \right].
\end{align}
Using the definition of $P_i(M)$, this becomes
\begin{align}
    \mathcal{L}
    =
    &
    \left(
        1-\sum_{i=1}^N w_i
    \right)t
    +
    \Tr Z
    \notag\\
    &
    +
    \sum_{a=1}^K
    \Tr\!\left[
        \left(
            q_a\sum_{i=1}^N
            w_i\sigma_{a,i}^{(n)}
            -Z
        \right)
        M_a
    \right].
\end{align}

The supremum over $t$ is finite only when
\begin{equation*}
    \sum_{i=1}^N w_i=1.
\end{equation*}
Similarly, the supremum over $M_a\succeq0$ is finite only when
\begin{equation*}
    Z\succeq
    q_a\sum_{i=1}^N
    w_i\sigma_{a,i}^{(n)},
    \qquad a=1,\ldots,K .
\end{equation*}
Hence the dual SDP is
\begin{align}
    s_{n,N}^\star
    =
    \inf_{w,Z}
    \quad&
    \Tr Z
    \label{eq:joint_sdp_lfp}\\
    \mathrm{s.t.}
    \quad&
    Z\succeq
    q_a\sum_{i=1}^N
    w_i\sigma_{a,i}^{(n)},
    \qquad a=1,\ldots,K,
    \notag\\
    &
    w_i\ge0,
    \qquad i=1,\ldots,N,
    \notag\\
    &
    \sum_{i=1}^N w_i=1,
    \notag\\
    &
    Z=Z^\dagger .
    \notag
\end{align}
Thus, the dual variables
\begin{equation*}
    w=(w_1,\ldots,w_N)
\end{equation*}
form a probability distribution on the nuisance-parameter grid.

For fixed $w$, the minimization over $Z$ in
Eq.~\eqref{eq:joint_sdp_lfp} is precisely the Bayes dual SDP
\eqref{eq:bayes_dual_sdp}, since
\begin{equation*}
    \bar\rho_a^{(n)}(w)
    =
    \sum_{i=1}^N
    w_i\sigma_{a,i}^{(n)}.
\end{equation*}
Therefore, the dual problem can equivalently be written as
\begin{equation*}
    \inf_{w\in\Delta_N} S_n(w).
\end{equation*}
Since
\begin{equation*}
    \beta_n(w)=1-S_n(w),
\end{equation*}
an optimal dual distribution $w^\star$ satisfies
\begin{equation*}
    w^\star
    \in
    \arg\max_{w\in\Delta_N}\beta_n(w).
\end{equation*}
Hence $w^\star$ is a least favorable prior on the finite grid.

Slater's condition also holds for
Eq.~\eqref{eq:minimax_primal_sdp}.
For example, one may choose $M_a=\id/K$ and take $t$
strictly smaller than $\min_i P_i(M)$.
Strong duality therefore gives
\begin{equation}
    V_{n,N}^{\rm succ}
    =
    t^\star
    =
    s_{n,N}^\star
    =
    \Tr Z^\star,
    \label{eq:minimax_strong_duality}
\end{equation}
and consequently
\begin{equation*}
    V_{n,N}^{\rm risk}
    =
    1-V_{n,N}^{\rm succ}.
\end{equation*}

Thus, the primal and dual SDPs provide, respectively, a finite-grid
minimax measurement $M^\star$ and a corresponding least favorable
prior $w^\star$.

%%%%%%%%%%%%%%%%%%%%%%%%%%%%%%%%%%%%%%%%%%%%%%%%%%%%%%%%%%%%%%%%%%%%%%%%%%

\subsection{Complementary slackness and numerical implementation}

The primal--dual formulation also provides a direct relation between
the support of the LFP and the active worst-case nuisance parameters.
Complementary slackness for the constraints $P_i(M)\ge t$ gives
\begin{equation}
    w_i^\star
    \left[
        P_i(M^\star)-V_{n,N}^{\rm succ}
    \right]
    =
    0,
    \qquad i=1,\ldots,N.
    \label{eq:complementary_lfp}
\end{equation}
Since
\begin{equation*}
    P_i(M^\star)
    \ge
    V_{n,N}^{\rm succ}
\end{equation*}
for every $i$, Eq.~\eqref{eq:complementary_lfp} implies
\begin{equation}
    w_i^\star>0
    \quad\Longrightarrow\quad
    P_i(M^\star)
    =
    V_{n,N}^{\rm succ}.
    \label{eq:lfp_active_constraint}
\end{equation}
Thus, every grid point carrying positive LFP weight is an active
worst case of the minimax measurement.
This is the SDP counterpart of the equalization property derived in
Sec.~\ref{sec:lfp_minimax}.

The matrix complementary-slackness conditions further imply
\begin{equation}
    \left[
        Z^\star
        -
        q_a\bar\rho_a^{(n)}(w^\star)
    \right]
    M_a^\star
    =
    0,
    \qquad a=1,\ldots,K.
    \label{eq:minimax_matrix_cs}
\end{equation}
Together with
\begin{equation*}
    Z^\star
    \succeq
    q_a\bar\rho_a^{(n)}(w^\star),
\end{equation*}
these are precisely the Holevo--Helstrom optimality conditions for
the nuisance-averaged ensemble associated with $w^\star$.
Therefore, the minimax measurement $M^\star$ obtained from the primal
SDP is simultaneously Bayes optimal for the corresponding LFP
$w^\star$, in agreement with Theorem~\ref{thm:lfp_minimax}.

In the numerical calculations of Sec.~\ref{sec:numerical_results},
we use the optimal primal POVM as
$M_{\rm LFP}^{\rm global}$ and the optimal dual weights as the
corresponding LFP $w^\star$.
Consequently, $M_{\rm LFP}^{\rm global}$ is minimax on the
optimization grid by construction, and no additional selection among
possibly nonunique Bayes-optimal measurements is required.

\bibliographystyle{apsrev4-2}
\bibliography{fujiki_refrence}

\end{document}